%% file: main.tex
\documentclass[conference]{IEEEtran}
\makeatletter
\def\ps@headings{%
\def\@oddhead{\mbox{}\scriptsize\rightmark \hfil \thepage}%
\def\@evenhead{\scriptsize\thepage \hfil \leftmark\mbox{}}%
\def\@oddfoot{}%
\def\@evenfoot{}}
\makeatother
\usepackage{amsmath}
\usepackage{graphicx}
\usepackage{cite}
\usepackage{url}

\usepackage{subfigure}
\usepackage{latexsym}
\usepackage{amsmath}
\usepackage{amsfonts}
\usepackage{epsfig}
\usepackage{color}
\usepackage{graphicx}
\usepackage{boxedminipage}

\newtheorem{theorem}{Theorem}[section]
\newtheorem{definition}[theorem]{Definition}

\newtheorem{lemma}[theorem]{Lemma}

\newtheorem{observation}[theorem]{Observation}
\newenvironment{proof}{\begin{trivlist}
\item[\hspace{\labelsep}{\bf\noindent Proof: }]}{\qedsymb\end{trivlist}}

\newcommand{\qedsymb}{\hfill{\rule{2mm}{2mm}}}

\newcommand{\remove}[1]{}
\def\N{{\cal N}}
\def\R{{\cal R}}

\def\p{{\mathbf{p}}}

\def\C{\widetilde C}

\def\paths{{\cal P}}

\begin{document}

\sloppy

%\pagestyle{plain}

%\begin{titlepage}

\title{The Impact of Exponential Utility Costs in Bottleneck Routing Games}

\author{
\IEEEauthorblockN{Rajgopal Kannan}
\IEEEauthorblockA{
Computer Science Department\\
Louisiana State University\\
279 Coates Hall\\
Baton Rouge, LA 70803, USA\\
Email: rkannan@csc.lsu.edu}
\and
\IEEEauthorblockN{Costas Busch}
\IEEEauthorblockA{
Computer Science Department\\
Louisiana State University\\
286 Coates Hall\\
Baton Rouge, LA 70803, USA\\
Email: busch@csc.lsu.edu}
}

\date{}

\maketitle

\begin{abstract}
We study bottleneck routing games where the
social cost is determined by the worst congestion
on any edge in the network.
Bottleneck games have been studied in the literature
by having the player's utility costs to be determined by the worst congested edge in
their paths.
However,
the Nash equilibria of such games are inefficient since the price of anarchy can
be very high with respect to the parameters of the game.
In order to obtain smaller price of anarchy we
explore {\em exponential bottleneck games} where the utility costs of the players
are exponential functions on the congestion of the edges in their paths.
We find that exponential bottleneck games are very efficient giving
a poly-log bound on the price of anarchy: $O(\log L \cdot \log |E|)$,
where $L$ is the largest path length in the players strategy sets and $E$
is the set of edges in the graph.
\end{abstract}

%\noindent
%{\bf Keywords:} Algorithmic game theory; Congestion game;
%Routing game; Nash equilibrium; Price of anarchy.

%\thispagestyle{empty}
%\end{titlepage}

\input{intro}
\input{definitions}
\input{stability}

\input{anarchy}

\section{Conclusions}
\label{section:conclusions}
In this work we have considered exponential bottleneck games
with player utility costs that are exponential functions
on the congestion of the edges of the players paths.
The social cost is $C$, the maximum congestion on any edge
in the graph.
We show that the price of anarchy is poly-log with respect to
the size of the game parameters: $O(\log L \cdot \log |E|)$,
where $L$ is the largest path length in the players strategy sets,
and $E$ is the set of edges in the graph.

\begin{figure}[ht]
\begin{center}
\resizebox{2.5in}{!}{\input{sumc.pstex_t}}
\end{center}
\caption{High price of anarchy with linear utility cost functions}
\label{fig:sumC}
\end{figure}
Several questions remain to be investigated in the future.
A natural question that arises is what is the impact of polynomial cost functions
to the price of anarchy.
Polynomial cost functions with low degree give high price of anarchy.
Consider the game instance of Figure \ref{fig:sumC}
where the player cost is $pc_i = \sum_{e \in p_i} C_e$
which is a linear function on the congestion of the edges
on the player's path.
In this game there $k$ players $\pi_1, \ldots, \pi_k$
where all the players have source $u$ and destination $v$
which are connected by edge $e$.
The graph consists of $k-1$ edge-disjoint paths from $u$ to $v$
each of length $k$.
There is a Nash equilibrium, depicted in the top of Figure \ref{fig:sumC},
where every player chooses to use a path of length 1 on edge $e$.
This is an equilibrium because the cost of each player is $k$,
while the cost of every alternative path is also $k$.
Since the congestion of edge $e$ is $k$ the social cost is $k$.
The optimal solution for the same routing problem is depicted in
the bottom of Figure \ref{fig:sumC},
where every player uses a edge-disjoint path
and thus the maximum congestion on any edge is 1.
Therefore, the price of anarchy is at least $k$.
Since we can choose $k = \Theta(\sqrt n)$, where $n$ is the number of nodes
in the graph,
the price of anarchy is $\Omega(\sqrt n)$.

%The above observation can be extended to low degree polynomial functions.
%A future problem is to explore polynomial cost functions of larger degrees.
%Also it is interesting to explore splittable games as in \cite{BO07}.
%Finally it is and interesting problem to study the effects of exponential cost functions
%in bicriteria routing games where the social cost metric is $C+D$
%(such games have been studied in \cite{BKV08,BKV09} using the regular $C_i+D_i$ player cost).

%\newpage
%\nocite{*}
{
\bibliographystyle{IEEEtran}
\bibliography{routing,oblivious,game}
}

\end{document}

%% file: intro.tex
\section{Introduction}
\label{section:intro}
Motivated by the selfish behavior of entities in communication
networks, we study routing games in general networks
where each packet's path is controlled independently by a selfish player.
We consider noncooperative games with $N$ players,
where each player has a {\em pure strategy profile}
from which it selfishly selects a single
path from a source node to a destination node
such that the selected path minimizes the player's utility cost function
(such games are also known as {\em atomic} or {\em unsplittable-flow} games).
We focus on {\em bottleneck games} where the objective for the social outcome
is to minimize $C$, the maximum congestion on any edge in the network.
Typically, the congestion on an edge is a non-decreasing function
on the number of paths that use the edge;
here, we consider the congestion to be simply the number of paths that use the edge.

Bottleneck congestion games have been studied
in the literature \cite{BO07,BM06,BKV08,BKV09}
where each player's utility cost is
the worst congestion on its path edges.
In particular, player $i$ has utility cost function
$\max_{e \in p_i} C_e$
where $p_i$ is the path of the player and
$C_e$ denotes the congestion of edge $e$.
In \cite{BO07} the authors observe that
bottleneck games are important in networks for various practical reasons.
In wireless networks the maximum congested edge is related to the
lifetime of the network since the nodes adjacent to high congestion edges
transmit large number of packets which results to
higher energy utilization.
Thus, minimizing the maximum edge congestion immediately translates
to longer network lifetime.
High congestion edges also result to congestion hot-spots in the network which
may slow down the performance of the whole network.
Hot spots may also increase the vulnerability of the network
to malicious attacks which aim to to increase the congestion of edges
in the hope to bring down the network or degrade its performance.
Thus, minimizing the maximum congested edge results to hot-spot avoidance
and also to more secure networks.

Bottleneck games are also important from a theoretical point of view
since the maximum edge congestion is immediately related to the
optimal packet scheduling.
In a seminal result,
Leighton {\it et al.} \cite{LMR94}
showed that there exist packet scheduling algorithms that
can deliver the packets along their chosen paths in time very close to $C+D$,
where $D$ is the maximum chosen path length.
This work on packet scheduling has been extended in \cite{CMS96,LMR94,LMR99,OR97,RT96}.
When $C \gg D$,
the congestion becomes the dominant factor in the packet scheduling performance.
Thus, smaller $C$ immediately implies faster delivery time for the
packets in the network.

A natural problem that arises concerns the effect of the players'
selfishness on the welfare of the whole network
measured with the {\em social cost} $C$.
We examine the consequence of the selfish behavior
in pure {\em Nash equilibria} which are stable states of the game in which no
player can unilaterally improve her situation.
We quantify the effect of
selfishness with the {\em price of anarchy}
($PoA$)~\cite{KP99,P01}, which expresses how much larger is the
worst social cost in a Nash equilibrium compared to the social cost
in the optimal coordinated solution.
The price of anarchy provides a measure for estimating
how closely do Nash equilibria of bottleneck routing games approximate the
optimal $C^*$ of the respective routing optimization problem.

Ideally, the price of anarchy should be small.
However, the findings in the literature show that bottleneck games are not
efficient, namely, the price of anarchy may be large.
In \cite{BO07} it is shown that if the edge-congestion function is bounded
by some polynomial with degree $d$ (with respect to the packets that use the edge)
then $PoA = O(|E|^d)$,
where $E$ is the set of edges in the graph.
In \cite{BM06} the authors consider the case $d=1$ and they show that $PoA = O(L + \log |V|)$,
where $L$ is the maximum path length in the players strategies and $V$ is the set of nodes.
This bound is asymptotically tight since there are game instances with $PoA = \Omega(L)$.
Note that $L \leq |E|$, and further $L$ may be significantly smaller than $|E|$.
However, $L$ can still be proportional to the size of the graph,
and thus the price of anarchy can be large.

\subsection{Contributions}

In this work we focus on exploring alternative utility
cost functions for the players
that have better impact on the social cost $C$.
We introduce {\em exponential bottleneck games}
where the player utilities are exponential functions
on the congestion of the edges of the paths.
In particular, the player utility cost function for player $i$ is:
$$\sum_{e \in p_i} 2^{C_e},$$
where $p_i$ is the player's chosen path.
Note that the new utility cost is a sum of exponential
terms on the congestion of the edges in the path
(instead of the max that we described earlier).
Using the new utility cost functions
we show that exponential games have always Nash equilibria
which can be obtained by best response dynamics.
The main result is that the price of anarchy is poly-log:
$$PoA = O(\log L \cdot \log |E|),$$
where $L$ is the maximum path length in the players strategy set
and $E$ is the set of edges in the graph.
This price of anarchy bound is a significant improvement over the
price of anarchy from the regular
utility cost functions described earlier.

Exponential cost functions are legitimate metrics for the utility costs of players
since they reflect the performance of the chosen paths according to congestion.
Each player is motivated to select a path
with lower utility cost
since it will provide a better quality path with lower congestion that can
affect positively the player's performance.
As we discuss in Section~\ref{section:conclusions},
the reason that we use exponential cost functions instead of polynomial
ones is that low degree polynomials give high price of anarchy.

\subsection{Related Work}

Congestion games were introduced and
studied in~\cite{monderer1,rosenthal1}.
Koutsoupias and Papadimitriou \cite{KP99} introduced
the notion of price of anarchy
in the specific {\em parallel link networks} model
in which they provide the bound $PoA = 3/2$.
Since then, many routing and congestion game models have been studied
which are distinguished by the network topology,
cost functions,
type of traffic (atomic or splittable),
and kind of equilibria (pure or mixed).
Roughgarden and Tardos \cite{roughgarden3}
provided the first result for splittable flows in general networks
in which they showed that $PoA\le 4/3$ for a player
cost which reflects to the sum of congestions of the edges of a path.
Pure equilibria with atomic flow have been studied in
\cite{BM06,CK05,libman1,rosenthal1,STZ04} (our work fits into this category),
and with splittable flow in
\cite{roughgarden1,roughgarden2,roughgarden3,roughgarden5}.
Mixed equilibria with atomic flow have been studied in
\cite{CKV02,czumaj1,FKP02,GLMMb04,GLMM04,GLMMR04,KMS02,KP99,LMMR04,MS01,P01},
and with splittable flow in
\cite{correa1,FKS02}.

Most of the work in the literature uses a cost metric
measured as the sum of congestions of all the edges of the player's path
\cite{CK05,GLMMR04,roughgarden2,roughgarden3,roughgarden5,STZ04}.
Our work differs from these approaches since we adopt
the exponential metric for player cost.
The vast majority of the work on routing games has been performed
for parallel link networks,
with only a few exceptions on general network topologies
\cite{BM06, CK05,correa1,roughgarden1},
which we consider here.

Our work is close to \cite{BM06},
where the authors consider the player cost $C_i$ and social cost $C$.
They prove that the price of stability is 1.
They show that the price of anarchy is bounded by $O(L + \log n)$,
where $L$ is the maximum allowed path length.
They also prove that $\kappa \leq PoA \leq c(\kappa^2 + log^2 n)$,
where $\kappa$ is the size of the largest edge-simple cycle in the graph and $c$ is a constant.
Some of the techniques that we use in our proofs
(for example expansion) were introduced in~\cite{BM06}.
Another related result for general
networks which has a brief discussion
of the convergence of maximum
player cost ($C_i$) games is \cite{libman1}
where the authors focus on parallel link networks, but
also give some results for general topologies on convergence to
equilibria.

Bottleneck congestion games have been studied in~\cite{BO07},
where the authors consider the maximum congestion metric in general networks
with splittable and atomic flow (but without considering path lengths).
They prove the existence and non-uniqueness
of equilibria in both the splittable and atomic
flow models.
They show that finding the best Nash equilibrium that
minimizes the social cost is a NP-hard problem.
Further, they show that the price of anarchy may be unbounded for specific
{\em edge congestion functions} (these are functions of the number of paths that use the edge).
If the edge congestion function is polynomial with degree $p$
then they bound the price of anarchy with $O(m^p)$,
where $m$ is the number of edges in the graph.
In the splittable case they show
that if the users always
follow paths with low congestion then the equilibrium
achieves optimal social cost.

%This observation is quantified in terms of the
%{\em price of stability} $(PoS)$~\cite{anshelevich2,anshelevich1}
%which expresses how much larger is the best social cost in a Nash equilibrium
%with respect to the social cost in the optimal coordinated solution.
%
%Let $C^*+D^*$ denote an optimal solution.
%In \cite{srinivasan1} it is shown that there is a constant
%approximation to this problem for general network topologies.
%Oblivious routing.

\subsection*{Outline of Paper}
In Section \ref{section:definitions}
we give basic definitions.
In section \ref{section:stability} we show that
exponential bottleneck games have always Nash equilibria.
We study the price of anarchy in Section \ref{section:anarchy}.
We finish with conclusions and future work in Section \ref{section:conclusions}.

%% file: definitions.tex
\section{Definitions}
\label{section:definitions}

\subsection{Path Routings}
Consider an arbitrary graph $G = (V,E)$ with nodes $V$ and edges $E$.
Let $\Pi = \{ \pi_1, \ldots, \pi_N \}$
be a set of packets such that each $\pi_i$ has
a source $u_i$ and destination $v_i$.
A {\em routing}
$\p=[p_1,p_2,\cdots,p_N]$ is a collection of paths,
where $p_i$ is a path for packet $\pi_i$ from $u_i$ to $v_i$.
We will denote by $E(p_i)$ the set of edges in path $p_i$.
Consider a particular routing $\p$.
The \emph{edge-congestion} of an edge $e$, denoted $C_e$,
is the number of paths in $\p$ that use edge $e$.
For any set of edges $A \subseteq E$,
we will denote by $C_{A} = \max_{e \in A} C_e$.
For any path $q$, the \emph{path-congestion} is $C_q = C_{E(q)}$.
For any path $p_i \in \p$,
we will also use the notation $C_i = C_{p_i}$.
The \emph{network congestion} is $C = C_{E}$,
which is the maximum edge-congestion over all edges in $E$.

We continue with definitions of exponential functions on congestion.
Consider a routing $\p$.
For any edge $e$, we will denote $\C_e = 2^{C_e}$.
For any set of edges $A \subseteq E$, we will denote $\C_{A} = \sum_{e \in A} \C_e$.
For any path $q$, we will denote $\C_q = \C_{E(q)}$.
For any path $p_i \in \p$ we will denote $\C_{i} = \C_{p_i}$.
We denote the length (number of edges) of any path $q$ as $|q|$.
%For any path $p_i \in \p$,
%we will also use the notation $D_{p_i} = D_i = |p_i|$.
%The maximum path length in routing $\p$ is $D = \max_{p \in \p} |p|$.
Whenever necessary we will append $(\p)$ in the above definitions
to signify the dependance on routing $\p$.
For example, we will write $C(\p)$ instead of $C$.

\subsection{Routing Games}
A routing game in graph $G$ is a tuple
$\R = (G,\N,\paths)$, where
$\N=\{1,2,\ldots,N\}$
is the set of players such
that each player $i$ corresponds to a packet $\pi_i$
with source $u_i$ and destination $v_i$,
and $\paths$ are the strategies of the players.
We will use the notation $\pi_i$ to denote player $i$
and its respective packet.
In the set $\paths = \bigcup_{i \in \N} \paths_i$
the subset $\paths_i$ denotes the {\em strategy set} of player $\pi_i$
which a collection of
available paths in $G$ for player $\pi_i$ from $u_i$ to $v_i$.
Any path $p \in \paths_i$
is a {\em pure strategy} available to player $\pi_i$.
A {\em pure strategy profile} is any routing $\p=[p_1,p_2,\cdots,p_N]$,
where $p_i \in \paths_i$.
The {\em longest path length} in $\paths$
is denoted $L(\paths) =\max_{p\in \paths}|p|$.
(When the context is clear we will simply write $L$).

For game $R$ and routing $\p$,
the \emph{social cost} (or {\em global cost})
is a function of routing $\p$, and it is
denoted $SC(\p)$.
The \emph{player or local cost} is also a function on $\p$
denoted $pc_i(\p)$.
We use the standard notation
$\p_{-i}$ to refer to the collection of paths
$\{p_1,\cdots,p_{i-1},p_{i+1},\cdots,p_N\}$, and
$(p_i;\p_{-i})$ as an alternative notation for $\p$ which
emphasizes the dependence on $p_i$.
Player $\pi_i$
is \emph{locally optimal} (or {\em stable}) in routing $\p$ if
$pc_i(\p) \leq pc_i(p_i';\p_{-i})$ for
all paths $p_i'\in\paths_i$.
A greedy move by a player $\pi_i$ is any change of its path from $p_i$ to $p'_i$
which improves the player's cost,
that is, $pc_i(\p) > pc_i(p'_i;\p_{-i})$.
{\em Best response dynamics} are sequences of greedy moves by players.

A routing $\p$ is in a Nash Equilibrium
(we say $\p$ is a \emph{Nash-routing})
if every player is locally optimal.
Nash-routings quantify the notion of
a stable selfish outcome.
In the games that we study there could exist multiple Nash-routings.
A routing $\p^*$ is an optimal pure strategy profile
if it has minimum
attainable social cost: for any other pure strategy profile
$\p$, $SC(\p^*)\le SC(\p)$.

We quantify the quality of the Nash-routings with the
\emph{price of anarchy} ($PoA$)
(sometimes referred to as the coordination ratio)
and the
\emph{price of stability} ($PoS$).
Let $\bf P$ denote the set of distinct
Nash-routings, and let
$SC^*$ denote the
social cost of an optimal routing $\p^*$.
Then,
\begin{equation*}
PoA
=\sup\limits_{\p\in ~{\bf P}} \frac{SC(\p)}{SC^*},
\qquad
PoS
= \inf\limits_{\p\in ~{\bf P}} \frac{SC(\p)}{SC^*}.
\end{equation*} 

%% file: stability.tex
\section{Exponential Bottleneck Games and their Stability}
\label{section:stability}

Let $\R = (G,\N,\paths)$
be a routing game
such that for any routing $\p$
the social cost function is $SC = C$,
and the player cost function is $pc_i = \C_i$.
We refer to such routing games as {\em exponential bottleneck games}.

We show that exponential games have always Nash-routings.
We also show that there are instances of exponential games
that have multiple Nash-routings.
The existence of Nash routings relies on finding
an appropriate potential function that provides an ordering of the routings.
Given an arbitrary initial state
a greedy move of a player can only
give a new routing with smaller order.
Thus, best response dynamics (repeated greedy moves)
converge to a routing where no player can improve further,
namely, they converge to a Nash-routing.
The potential function that we will use is:
$f(\p) = \C_{E}(\p)$.
We show that any greedy move gives a new routing with lower potential.

\begin{lemma}
\label{theorem:greedy}
If in routing $\p$ a player $\pi_i$ performs a greedy move,
then the resulting routing $\p'$ has $\C_E(\p) > \C_E(\p')$.
\end{lemma}

\begin{proof}
Suppose that player $\pi_i$ has path $p_i \in \p$ and switches to path $p'_i \in \p'$.
Then, $\C_{p_i}(\p) > \C_{p'_i}(\p')$.
Let $A = E(p_i) \setminus E(p'_i)$ and $B = E(p'_i) \setminus E(p_i)$.
It has to be that
$\C_{A}(\p) > \C_{B}(\p')$
since $\pi_i$'s cost decreases.
Further, $\C_{B}(\p') = 2 \C_{B}(\p)$ and $\C_{A}(\p) = 2 \C_{A}(\p')$,
since the presence or absence of player's $\pi$ path in the edges $A$ and $B$ alters their
total cost by a factor of 2.
Let $H = E \setminus\{A \cup B\}$.
We have that $\C_{H}(\p) = \C_{H}(\p')$, since $\pi_i$ does not affect those edges.
Since $E = H \cup A \cup B$ and $H$, $A$, $B$ are disjoint,
we have that
\begin{eqnarray*}
\C_{E}(\p)
& = & \C_{H}(\p) + \C_A(\p) + \C_B(\p) \\
& = & \C_{H}(\p') + 2 \C_A(\p') + \frac{\C_B(\p')} {2} \\
& = & \C_{H}(\p') + \C_A(\p') + \C_B(\p') \\
&&    + \left ( \C_A(\p') - \frac{\C_B(\p')} {2} \right ) \\
& = & \C_E(\p') + \left ( \C_A(\p') - \frac{\C_B(\p')} {2} \right ).
\end{eqnarray*}
Since $\C_{A}(\p) = 2 \C_{A}(\p')$ and $\C_{A}(\p) > \C_{B}(\p')$,
we have that $2 \C_A(\p') > \C_B(\p')$,
or equivalently
$$
\C_A(\p') - \frac{\C_B(\p')} {2} > 0.
$$
Therefore, $\C_{E}(\p) > \C_E(\p')$, as needed.
\end{proof}

Since the result of the potential function cannot be smaller than zero,
Lemma \ref{theorem:greedy} implies that best response dynamics converge to Nash-routings.
Thus, we have:

\begin{theorem}
\label{theorem:stability}
Every exponential game instance $\R = (G,\N,\paths)$ has a Nash-routing.
\end{theorem}

\begin{figure}[ht]
\begin{center}
\resizebox{2.5in}{!}{\input{multiple-nash.pstex_t}}
\end{center}
\caption{An exponential game instance with multiple Nash-routings}
\label{fig:multiple-nash}
\end{figure}

We continue to show that there are exponential games with multiple Nash-routings.
Consider the example of Figure \ref{fig:multiple-nash}.
There are three players $\pi_1, \pi_2, \pi_3$ with respective sources $u_1, u_2, u_3$
and destinations $v_1, v_2, v_3$.
The strategy set of each player are all feasible paths from their source to destination.
In the left part of Figure \ref{fig:multiple-nash} is a Nash-routing $\p = [p_1, p_2, p_3]$ with
social cost $SC(\p) = 2$ and respective player costs $pc_1(\p) = 4$, $pc_2(\p) = 8$, and $pc_3(\p)=6$.
On the right part of the same figure is another Nash-routing $\p' = [p'_1, p'_2, p'_3]$ with
social cost $SC(\p') = 1$ and respective player costs $pc_1(\p') = 2$, $pc_2(\p')=6$, and  $pc_3(\p') = 6$.
Thus, we can make the following observation:

\begin{observation}
There exist exponential game instances with multiple Nash-routings.
\end{observation}

%% file: multiple-nash.pstex_t
\begin{picture}(0,0)%
\includegraphics{multiple-nash.pstex}%
\end{picture}%
\setlength{\unitlength}{3947sp}%
\begingroup\makeatletter\ifx\SetFigFont\undefined%
\gdef\SetFigFont#1#2#3#4#5{%
  \reset@font\fontsize{#1}{#2pt}%
  \fontfamily{#3}\fontseries{#4}\fontshape{#5}%
  \selectfont}%
\fi\endgroup%
\begin{picture}(5341,2659)(2046,-5963)
\put(5789,-3681){\makebox(0,0)[rb]{\smash{{\SetFigFont{12}{14.4}{\rmdefault}{\mddefault}{\updefault}$u_1$}}}}
\put(5310,-4141){\makebox(0,0)[rb]{\smash{{\SetFigFont{12}{14.4}{\rmdefault}{\mddefault}{\updefault}$u_2$}}}}
\put(5338,-5081){\makebox(0,0)[rb]{\smash{{\SetFigFont{12}{14.4}{\rmdefault}{\mddefault}{\updefault}$u_3$}}}}
\put(6502,-3691){\makebox(0,0)[lb]{\smash{{\SetFigFont{12}{14.4}{\rmdefault}{\mddefault}{\updefault}$v_1$}}}}
\put(6944,-4141){\makebox(0,0)[lb]{\smash{{\SetFigFont{12}{14.4}{\rmdefault}{\mddefault}{\updefault}$v_2$}}}}
\put(6934,-5081){\makebox(0,0)[lb]{\smash{{\SetFigFont{12}{14.4}{\rmdefault}{\mddefault}{\updefault}$v_3$}}}}
\put(2969,-3678){\makebox(0,0)[rb]{\smash{{\SetFigFont{12}{14.4}{\rmdefault}{\mddefault}{\updefault}$u_1$}}}}
\put(2490,-4138){\makebox(0,0)[rb]{\smash{{\SetFigFont{12}{14.4}{\rmdefault}{\mddefault}{\updefault}$u_2$}}}}
\put(2518,-5078){\makebox(0,0)[rb]{\smash{{\SetFigFont{12}{14.4}{\rmdefault}{\mddefault}{\updefault}$u_3$}}}}
\put(3682,-3688){\makebox(0,0)[lb]{\smash{{\SetFigFont{12}{14.4}{\rmdefault}{\mddefault}{\updefault}$v_1$}}}}
\put(4124,-4138){\makebox(0,0)[lb]{\smash{{\SetFigFont{12}{14.4}{\rmdefault}{\mddefault}{\updefault}$v_2$}}}}
\put(4114,-5078){\makebox(0,0)[lb]{\smash{{\SetFigFont{12}{14.4}{\rmdefault}{\mddefault}{\updefault}$v_3$}}}}
\put(6134,-4467){\makebox(0,0)[b]{\smash{{\SetFigFont{12}{14.4}{\rmdefault}{\mddefault}{\updefault}$p'_2$}}}}
\put(6152,-5425){\makebox(0,0)[b]{\smash{{\SetFigFont{12}{14.4}{\rmdefault}{\mddefault}{\updefault}$p'_3$}}}}
\put(6143,-3462){\makebox(0,0)[b]{\smash{{\SetFigFont{12}{14.4}{\rmdefault}{\mddefault}{\updefault}$p'_1$}}}}
\put(3326,-3885){\makebox(0,0)[b]{\smash{{\SetFigFont{12}{14.4}{\rmdefault}{\mddefault}{\updefault}$p_2$}}}}
\put(3307,-4862){\makebox(0,0)[b]{\smash{{\SetFigFont{12}{14.4}{\rmdefault}{\mddefault}{\updefault}$p_3$}}}}
\put(3326,-3472){\makebox(0,0)[b]{\smash{{\SetFigFont{12}{14.4}{\rmdefault}{\mddefault}{\updefault}$p_1$}}}}
\put(3307,-5885){\makebox(0,0)[b]{\smash{{\SetFigFont{12}{14.4}{\rmdefault}{\mddefault}{\updefault}Nash-routing $\p$}}}}
\put(6143,-5904){\makebox(0,0)[b]{\smash{{\SetFigFont{12}{14.4}{\rmdefault}{\mddefault}{\updefault}Nash-routing $\p'$}}}}
\end{picture}%

%% file: anarchy.tex
\newcommand{\no}{\noindent}
\newcommand{\bea}{\begin{eqnarray}}
\newcommand{\eea}{\end{eqnarray}}
\newcommand{\beq}{\begin{equation}}
\newcommand{\eeq}{\end{equation}}
\newcommand{\beqs}{\begin{equation*}}
\newcommand{\eeqs}{\end{equation*}}
\newcommand{\beas}{\begin{eqnarray*}}
\newcommand{\eeas}{\end{eqnarray*}}
\newcommand{\ep}{\end{proof}}
\newcommand{\bp}{\begin{proof}}

\def\fct#1{{\mathop{\rm #1}}}   % general macro for functions font
\def\abs{\fct{abs}}             % absolute value function
\def\sign{\fct{sign}}           % sign function
\def\re{\fct{Re\,}}             % real part; \Re gives old-fashioned Re
\def\im{\fct{Im\,}}             % imag. part; \Im gives old-fashioned Im
\def\argmax{\mathop{\rm argmax}}% argmax; indices appear below word
\def\argmin{\mathop{\rm argmin}}% argmin; indices appear below word

\def\h{\hat{h}}
\def\hG{\hat{G}}
\def\hg{\hat{G}}
\def\hC{\hat{C}}
\def\hc{\hat{C}}
\def\hl{\frac{1}{2} \log_2 }
\def\PBS0{\frac{\hat{h}_b P_T}{M} }
\def\PG{P^g  }
\def\PB{P^b  }
\def\PGS{P^{g^*}  }
\def\PBS{P^{b^*}  }
\def\RG{R^g  }
\def\RB{R^b  }
\def\a{\alpha }
\def\DH{\Delta \! H }

\def\-{  \!-\! }
\def\+{  \!+\! }
\def\argmax{\operatornamewithlimits{arg\,max}}
\def\mbf{\mathbf}
\def\hsA{\hspace*{0.3in}}
\def\hsB{\hspace{0.5cm}}
\def\f{\frac{1}{2}}
\def\o{\overline}
\def\n{\nonumber}
\def\b{\bar}
\def\hP{\hat{P}}
\def\mb{\mbox}

\section{Price of Anarchy}
\label{section:anarchy}
We bound the price of anarchy in exponential bottleneck games.
Consider an exponential bottleneck routing game $\R = (G,\N,\paths)$.
Let $\p = [p_1, \ldots, p_N]$ be an arbitrary Nash-routing with social cost $\hC$;
from Theorem \ref{theorem:stability} we know that $\p$ exists.
Let $\p^* = [p^*_1, \ldots, p^*_N]$ represent the routing with optimal social cost $C^*$.
Let $L$ be the maximum path length in the players strategy sets
and $L^* \leq L$ be the longest path in $\p^*$.
Denote $l^* = \lg L^*$.

We will obtain an upper bound on the price of anarchy
$PoA = C/C^*$ by finding a lower
bound on the number of players as well as the number of edges in the $\p$.
The proof relies on the notion of self-sufficient sets:

\begin{definition}[Self-sufficient set]
{\it Consider an arbitrary set of players $S$ in Nash-routing $\p$ in game $\R$.
We label the equilibrium of $S$ as {\em self-sufficient} if,
after removing the paths of all players
$\notin S$ from $\p$, for every $\pi_i \in S$, the cost $\C_{p_i^*}$ remains at least $pc_i(\p)$.
Thus $\pi_i$ cannot switch to path $p_i^*$ only because of other players in $S$.
}
\end{definition}

If a set $S$ is not self-sufficient,
then additional players $S'$ must be present to guarantee the Nash-routing.
Thus, we define the notion of support sets:

\begin{definition}[Support set]
{\it
If $S$ is not a self-sufficient set,
then there is a set of players $S'$, where $S \cap S' = \emptyset$,
such that the paths of the players in $S' \cup S$
guarantee that $\C_{p_i^*}$ remains at least $pc_i(\p)$
even if all the other players $\notin S' \cup S$
are removed from the game.
}
\end{definition}

The players in $S'$ may not be self-sufficient either.
This process is repeated until a self-sufficient set is found.
Our goal is to find a lower bound on a self-sufficient set players.
We start with a small set of players based on $\hat{C}$
and the optimal congestion value $C^*$, prove they are not
self-sufficient and consider a sequence of {\it expansions} that will
eventually lead to a self-sufficient set.  We find the minimum
number of these expansions to terminate the process and thus find
the minimum number of players (and edges) needed to support a maximum
equilibrium congestion of $\hat{C}$.  For a given graph $G$ and
players/edges this gives us an upper bound on $\hat{C}$ relative
to $C^*$.

Initially assume $C^*=1$, i.e every player in the optimally congested
network has a unique optimal path to its destination of length at
most $L^*$.  For the game $\hg$ we will consider sets of players
in stages, depending on their costs in $\hg$.  Let $S^{(i)}$ denote
the set of players in stage $i$, $1 \leq i \leq \hc$ with player
costs $\tilde{C}: 2^{\hc \- i}  \+ 2 \leq \tilde{C} \leq 2^{\hc \-
i \+ 1}$.  Consider an arbitrary player $\pi$ in stage $i$.  We let
$P^*$ denote its optimal path and $\Phi(P^*)$ the minimum cost of
path $P^*$ in $\hg$. Since $\Pi$ is in equilibrium, we must have
$\Phi(P^*) > 2^{\hc \- i \- 1}$.

We formally define {\it expansion chains} as follows: In stage $i$,
$1 \leq i \leq \hc \- 1$, let $A^{(i)}$ denote the set of all players
occupying exactly one edge of congestion $\hc \- i \+ 1$,  let
$B^{(i)}$ denote the set of all players whose {\it maximum} edge
congestion $C'$ satisfies $\hc \- i \geq C' >  \hc \- i \- l^* \-
1$ and finally let $D^{(i)} = S^{(i)} \- A^{(i)} \-  B^{(i)} $.
For $i>1$, a level $i$ expansion chain consists of a single chain
of nodes $r \rightarrow X_{i+1}(r) \rightarrow X_{i+2}(r)\rightarrow \ldots$,
where the root node $r$ represents the players
of $\{ B^{(i)}, D^{(i)} \}$.  Thus there are two possible expansion
chains rooted at level $i$, except for level 1, where $A^{(1)}$ can
also be the root node for a third expansion chain.  The rest of the
chain consists of a sequence of nodes such that node $X_{i+k}(r)$
represents the support set of players of node $X_{i+k-1}(r)$.

%on all the optimal path edges (labeled
%expansion edges)

%Expansion chains bound the number of edges in the graph. As long
%as expansion chains rooted at levels $1,2, \ldots, k$ exist, the
%equilibrium of players at those congestion/cost levels is not
%self-sufficient and a support set of lower congestion (i.e more edges)
%must be present in the graph. The minimum number of edges in the
%graph can therefore be related to the terminating condition for
%expansion.

%As we will show later, expansion chains of the smallest
%size are composed only of type $A$ players, i.e all non-root nodes
%$X_{i+k}(r)$ of the chain are players of type $A$.  This implies
%that each succeeding node of an expansion chain consists of players
%on single edges of decreasing congestion.

We first show below, a sufficient condition on $\hC$ for expansion
chains to exist at any stage. For technical reasons, we will use
$l^*_1 = \log_2 (L^* \- 1)$.

\begin{lemma}
{\it
Given a non-empty player set $X^{(i)} \in \{ A^{(i)}, B^{(i)},
D^{(i)} \}$, either there exists an expansion chain rooted at $X^{(i)}$
or the players of $X^{(i)}$ are on the expansion chain of other
players for all stages $i: 1 \leq i \leq \hC \- l_1^* \- 11$.
\label{l^*lemma}
}
\end{lemma}
\begin{IEEEproof}
To prove the existence of expansion chains at any stage $i$, we
need to show that the set of players $X^{(i)}$ is not
self-sufficient. Consider each of the possible elements
of $X^{(i)}$ separately. First consider the set $D^{(i)}$. Clearly,
$D^{(i)}$'s equilibrium is not self-sufficient since the maximum
congestion experienced by players in $D^{(i)}$ is $\hC \- i \- l_1^*
\- 2$ and thus the maximum cost of an optimal path composed exclusively
of edges from $D^{(i)}$ is $ (L^*\-1\+1) \cdot 2^{\hC -i - l_1^* - 2} \ \
= \ \ 2^{\hC - i  - 2} + 2^{\hC - i  - l_1^* -2} $, which is
strictly less than the minimum required cost of an optimal path
$\Phi(P^*) \+ 1$.

Next consider the set $B^{(i)}$. Assume for purposes of contradiction
that $B^{(i)}$ is self-sufficient, i.e there are a
sufficient number of edges composed exclusively of players in
$B^{(i)}$ that are also on all the optimal paths of $B^{(i)}$ and
each optimal path has cost at least $\Phi(P^*) \+ 1$.  Let $B^{(i)}_j$
denote the edges of congestion $\hc \- i \- j$ composed exclusively
of players in $B^{(i)}$, where $ 0 \leq j \leq \hC \- i \-1 $.  Note
that a single player in $B^{(i)}$ may have several edges across
different $B^{(i)}_j$'s.  Each edge of $B^{(i)}_j$ contributes
$2^{\hc \- i \- j}$ to the total cost of each of the $\hc \- i \-
j$ players on the edge.  Since the total cost of each player in
$B^{(i)}$ is bounded by $2^{\hc \- i \+ 1}$, we must have

\bea
\; \; \sum_{j=0}^{\hc \- i \- 1} |B^{(i)}_j| (\hc \- i \-
j) 2^{\hc \- i \- j}  &\leq &|B^{(i)}| 2^{\hc \- i \+ 1}  \nonumber
\\
\equiv \sum_{j=0}^{l_1^* + 2} |B^{(i)}_j| \left(
\dfrac{\hc \- i \- j}{2^{j \+ 1}} \right)  &\leq &|B^{(i)}|
\label{lowerbound1}
\eea

Since $B^{(i)}$ is in equilibrium, each of the $|B^{(i)}|$ optimal
paths has cost $> \Phi(P^*)$. For $j \geq 1$, each edge $e \in B^{(i)}_j$ on an
optimal path $P^*_{opt}$ contributes $\Phi(P^*)/2^{j-1}$ towards the
cost of this path. (Each edge in $B^{(i)}_0$ contributes
$\Phi(P^*)$).  Now using the fact that $B^{(i)}$'s equilibrium is
self-contained, we must have

\bea
\sum_{e \in B^{(i)}_0} \Phi(P^*)
 +
\sum_{j=1}^{\hc \- i \- 1}
\sum_{e \in B^{(i)}_j} \dfrac{\Phi(P^*)}{2^{j-1}}
&>
& \sum_{P^*_{opt}} \Phi(P^*) \nonumber \\
\ \ \equiv
|B^{(i)}_0| + \sum_{j=1}^{\hc \- i \- 1} \dfrac{|B^{(i)}_j|}{2^{j-1}}
&>
&|B^{(i)}|
\label{upperbound1}
\eea

We note the following: edges of congestion $\leq \hC \- i \- l_1^* \- 3$ must
account for less than half the cost of any optimal path on which they are
present. The maximum contribution of such edges over $L^* \- 1$ edges of the
optimal path is $\Phi(P^*)/2$, implying that there must be one edge of higher
congestion ($\geq \hC \- i \- l_1^* \- 2$) that contributes more than half of
the required total cost $\geq \Phi(P^*)\+ 1$. Thus we must have
\beq
\sum_{j= l^* + 3 }^{\hc \- i \- 1} \dfrac{|B^{(i)}_j|}{2^{j-1}}
<
|B^{(i)}_0| +
\sum_{j=1}^{l_1^* + 2} \dfrac{|B^{(i)}_j|}{2^{j-1}}
\label{upperbound2}
\eeq
\no and therefore Eq.~\ref{upperbound1} becomes
\beq
|B^{(i)}_0| +
\sum_{j=1}^{l_1^* \+ 2} \dfrac{|B^{(i)}_j|}{2^{j-1}}
>
\dfrac{|B^{(i)}|}{2}
\label{upperbound3}
\eeq

Comparing Eq.~\ref{upperbound3} with Eq.~\ref{lowerbound1}, we get
\[
2|B^{(i)}_0| +
\sum_{j=1}^{l_1^* + 2} \dfrac{|B^{(i)}_j|}{2^{j-2}}
\geq
\sum_{j=0}^{l_1^* + 2} \dfrac{|B^{(i)}_j|}{2^{j+1}} \big(\hC \- i \- j \big)
\]
or simplifying
\beq
\sum_{j=0}^{l_1^* + 2} \dfrac{|B^{(i)}_j|}{2^{j+1}}
\cdot
\big( 8 \+ j \- (\hC \- i) \big)
\geq  0
\label{finalbound}
\eeq
Since $|B^{(i)}_j| > 0$ for at least some $j: 1 \leq j \leq l_1^* \+
2$, Eq.~\ref{finalbound} is impossible for $ (\hC \- i) > l_1^* \+ 10$, which contradicts
the assumption that $B^{(i)}$ is self-sufficient.

Finally for the case of players from $A^{(i)}$, each subset of $\hC
\- i \+ 1$ players shares an edge. Thus the maximum number of optimal
edges available from within the set is $|A^{(i)}|/(\hC \- i \+ 1)$.
Since this is much less than the number of optimal paths $|A^{(i)}|$,
players in $A^{(i)}$ are also not self-sufficient.

Concluding, none of the player sets $\{ A^{(i)}, B^{(i)}, D^{(i)}\}$ are
self-sufficient and hence either these players are on the expansion chains of
some other players or there are expansion chains rooted at these players
in stage $i: 1 \leq i \leq \hC \- l^* \- 11$.
\end{IEEEproof}

The above lemma guarantees the existence of at least one expansion
chain rooted at stage $1$ when $\hC = O(l^*)$.  We now want to find
the minimum number of edges required to support the game with
equilibrium cost $2^{\hC}$. This corresponds to finding the smallest
expansion chain rooted at stage $1$.
By our definition, an expansion chain consists of new
players occupying the expansion edges of players on the previous
levels.  It would seem that chains should consist of type $B$ players
since they occupy multiple edges and thus fewer players are required.
However as the lemma below shows it is players of type $A$ that
minimize the expansion edges.

%Consider the optimal path edges $X^*$ of the set of
%players $X$ at some level $i$ of an expansion chain. Label these
%optimal edges as expansion edges.  The player set $Y$ occupying
%edges $X^*$ will in turn have their own distinct optimal paths and
%expansion edges $Y^*$ (since $C^* =1$).  We would like to iteratively
%find the optimal set of players $X,Y$ that will have the minimal
%number of expansion edges $X^*,Y^*$. It would seem that the minimum
%number $Y^*$ can be obtained if previous stage expansion edges $X^*$
%were occupied by players of type $B$, since they occupy multiple
%edges and thus fewer players are required. However as the lemma
%below shows it is players of type $A$ that minimize the expansion edges.

%with the same maximum congestion $\hC$

Consider an arbitrary player $\pi$ of type $B$ in $\hG$ occupying
edges $E= \{e_1,e_2, \ldots, e_k\}$  of non-increasing congestion
$c_1 \geq c_2 \geq \ldots c_k$ that are optimal edges (expansion
edges) of other players, where we assume maximum congestion $c_1
\geq 2$.  We want to answer the following question: Is there an
alternate equilibrium/game containing player(s) with the same total
equilibrium cost as $\pi$, but requiring fewer edges to support
this equilibrium cost.  Note that when comparing these two games,
the actual routing paths (i.e source-destinations) do not have to
be the same. All we need to show is the existence of an alternate
game (even with different source-destination pairs for the players)
that has the same equilibrium cost.

In particular, consider an alternate game $G'$ in which $\pi$ is
replaced by a set $P=\{\pi_1, \pi_2, \ldots, \pi_k\}$ of type
$A$ players occupying single edges of congestion $c_1, c_2, \ldots,
c_k$, where $\pi$ and the set $P$ are also in equilibrium in
their respective games. The equilibrium cost of $\pi$ and set $P$
is the same ($\sum_{j=1}^k 2^{c_j}$) as they are occupying edges
of the same congestion.  Since both $\pi$ in game $\hG$ and the set
of players $P$ are in equilibrium and occupying expansion edges of
other players in their respective games, $C^*=1$ implies they must
have their own expansion edges in their respective games.  Suppose
we can show that the number of expansion edges required by the $k$
players in $P$ is at most those required by the single player of
type $B$.  Since $\pi$ is an arbitrary type $B$ player, this argument
applied recursively implies that all expansion edges in the game
$\hG$ should be occupied by type $A$ players to minimize the total
number of expansion edges.  Thus we will have shown that any
equilibrium with cost $\hC$ can be supported with fewer total players
if they are of type $A$ than if they are of type $B$.  Let $\pi^*$
and $P^*$ denote the expansion edges of $\pi$ and the set $P$
respectively.

\begin{lemma}
{\it
$|P^*| \leq |\pi^*|$ for arbitrary players $\pi$ and set $P$ with
the same equilibrium cost.
\label{typeB=typeA}
}
\end{lemma}
\begin{IEEEproof}
We prove this by strong induction on the length of player $\pi$'s
path.  For the basis, assume player $\pi$ is on path  $(e_1,e_2) $
of length 2 in $\hG$, with edges of congestion $c_1$ and $c_2$
respectively, where $c_1 \geq c_2$.  Simultaneously consider two
players $\pi_1$ and $\pi_2$ on single edges in game $G'$ with respective
costs $2^{c_1}$ and $2^{c_2}$.  We need to show that every possible
optimal path (i.e expansion edges) for $\pi$ in $\hG$ has two
equivalent optimal paths (of the same or lower total cost) for the
two players $\pi_1$ and $\pi_2$ in $G'$.

Suppose the optimal path of $\pi$ is $\pi^* = (e_1^*, e_2^*, \ldots
e_m^*)$ in non-increasing order of congestion $c_1^* \geq c_2^*
\geq \ldots c^*_m$.  Consider two cases:

%\begin{enumerate}

{\it Case 1 $c_1^* < c_1$}: Since $\pi$ is in equilibrium, $\sum_{i=1}^m
2^{c^*_i} \geq (2^{c_1} + 2^{c_2})/2$.  Since $c^*_1 < c_1$, there
exists $c^*_j$ such that $\sum_{i=1}^j 2^{c^*_i} = 2^{c_1 - 1}$.
Hence optimal path $\pi^*$ can be partitioned into two paths,
$\pi^*_1 = ( e^*_1, \ldots, e^*_j)$ and $\pi_2^* = (e^*_{j+1} \ldots
e^*_m)$ with costs $C(\pi^*_1) = 2^{c_1-1}$ and $C(\pi^*_2) \geq
2^{c_2-1}$. Thus the edges of $\pi_1^*$ and $\pi^*_2$ can serve as
expansion edges for $\pi_1$ and and $\pi_2$ in alternate game $G'$
with appropriate endpoints, specifically, the endpoints of $\pi^*_1$
and $\pi^*_2$ will be the same as the endpoints of edge $e_1$ and
$e_2$ in $G'$.  Hence $|P^*| = |\pi^*|$ in this case as desired.

{\it Case 2 $c_1^* \geq c_1$}: There are at least $c_1 \geq 2$ players
on player $\pi$'s optimal path with costs $\geq 2^{c_1}$. Since
$C^*=1$, these players must have independent optimal paths of cost
$\geq 2^{c_1-1}$. Hence at least $c_1 \geq 2$ such optimal paths
are needed to support $\pi$ in game $\hG$. In contrast, in game
$G'$, the two players $\pi_1$ and $\pi_2$ can be supported by
two edges of congestion $c_1 \- 1$ and $c_2 \- 1$, respectively.
Hence $|P^*| =2  \leq |\pi^*|$ in this case as well.

%\end{enumerate}

For the inductive hypothesis assume $|P^*| \leq |\pi^*|$ for all
paths upto length $k>2$. Consider player $\pi$ occupying edges of
non-increasing congestion $c_1, \ldots, c_{k+1}$ in $\hG$ whose
optimal path has edges of non-increasing congestion $c_1^*, \ldots,
c_m^*$. As before consider two cases,
{\it Case 1 $c_1^* < c_1$}:
let $j_1$ and $j_2$ be the indices such that
1) $\sum_{i=1}^{j_1} 2^{c_i^*} = (2^{c_1} + 2^{c_2})/2$,
and
2) $\sum_{i=j_1}^{j_2} 2^{c_i^*} = (\sum_{i=3}^{k+1} 2^{c_i})/2$.
Note that since $c_1^* < c_1$, indices $j_1$ and $j_2$ exist with
$j_1 < j_2 \leq m$.  Instead of player $\pi$ consider two new players
$P_1$ and $P_2$, where $P_1$ occupies two edges of congestion $c_1$
and $c_2$ and $P_2$ occupies edges of congestion $c_3, c_4 \ldots
c_{k+1}$.  From above, $j_2$ edges are required to satisfy $P_1$
and $P_2$ and $|\pi^*| =m \geq j_2$.  Players $P_1$ and $P_2$ have
path lengths $<k$ and thus by the inductive hypothesis, the number
of expansion edges $P^*$ required to support $P_1$ and $P_2$ assuming
they were replaced by type $A$ players satisfies $|P^*| \leq j_2
\leq |\pi^*|$ as desired.

{\it Case 2 $c_1^* \geq c_1$}: First assume $m \geq 2$. Let $j$ be the
{\it largest} index such that $\sum_{i=1}^{j} 2^{c_i} \leq 2^{c_1^*
+ 1}$.  Clearly $j$ exists since $c_1^* \geq c_1$. Now instead of
player $\pi$, consider two players $P_1$ and $P_2$ with $P_1$
occupying edges of congestion $c_1, c_2, \ldots c_j$ and $P_2$
occupying edges of congestion $c_{j+1}, \ldots c_m$, respectively.
The edge of congestion $c_1^*$ can satisfy $P_1$ while the remaining
edges of the optimal path $\pi^*$ can satisfy $P_2$.  As in the
previous case, players $P_1$ and $P_2$ have path lengths $<k$ and
thus by the inductive hypothesis, the number of expansion edges
$P^*$ required to support $P_1$ and $P_2$ assuming they were replaced
by type $A$ players satisfies $|P^*| \leq |\pi^*|$ as desired.
The case when $m=1$ is omitted for brevity.

%Now consider the case when when $m=1$.  First since $C^*=1$, the
%single edge of congestion $c^*_1$ on the optimal path of $\pi$ must
%result in least $c_1^*$ paths of cost $\geq 2^{c_1^*-1}$. Secondly,
%we must have
%\beq
%2^{c^*_1+1} \geq \sum_{i=1}^{k+1} 2^{c_i}
%\label{eq:case2m=1}
%\eeq
%equivalently $4 \cdot 2^{c_1^*-1} \geq \sum_{i=1}^{k+1} 2^{c_i}$.
%Hence if $\pi$ can considered as 4 players of cost at most $2^{c_1^*}$
%each (which is feasible, given the above), 4 paths of cost $2^{c_1^*-1}$
%are sufficient to satisfy these 4 players.  There are $c_1^*$ such
%paths available and we just need to show $c_1^* \geq 4$.  Since
%player $\pi$ is of length $k+1 \geq 3$ and $c_1 \geq 2$, using
%Eq.~\ref{eq:case2m=1}, we obtain that $c_1^* \geq 4$.  Now applying
%the inductive hypothesis to these 4 players, we get the result of
%the lemma.

\end{IEEEproof}

As a consequence of lemma~\ref{typeB=typeA},  we have
\begin{lemma}
{\it For $\hC > l^* \+ 11$, the expansion chain rooted in stage 1 and
occupying the minimum number of edges consists only of players of
type $A$ (other than the root).
}
\label{onlyAlemma}
\end{lemma}
%\begin{IEEEproof}
%Let $X$ be the set of players  with cost $2^{\hC}$ at the root of
%expansion chain $EC$ in stage 1. (Note that $|X| > 0$ by our original
%assumption). Excluding the edges of $X$, assume there are $k_j$
%edges of congestion $\hC \- j$ that form the the expansion edges
%(optimal paths) of $X$, where $j \leq l^*$.  Since all the players
%on these expansion edges are themselves in equilibrium and $C^*=1$,
%they require their own expansion edges.  By lemma~\ref{typeB=typeA},
%the number of these new expansion edges will be minimized if these
%players are all of type $A$ (occupying single edges of congestion
%$\hC \- j$) as opposed to type $B$ (occupying multiple such edges).
%As long as $\hC -\ j > l^* \+ 11$, players on these new expansion
%edges will be on $EC$ and will themselves require expansion. By
%reapplying lemma~\ref{typeB=typeA}, the number of these new expansion
%edges at each stage will be minimized with type $A$ players.
%\end{IEEEproof}
%

%%%%%%%%%%%%%%BEGIN NEW WORK%%%%%%%%%%%%%%%%%%%%%%%%%%%%

%%%%%%%%%%%%%%BEGIN NEW WORK%%%%%%%%%%%%%%%%%%%%%%%%%%%%

Next we derive the size of the smallest network required to support
an equilibrium congestion of $\hC$. Without loss of generality, we
assume there exists at least one type $A$ player in stage 1, i.e a
single edge of congestion $\hC$ and derive the minimum chain rooted
at $A^{(1)}$.  From lemma~\ref{onlyAlemma}, there exists an expansion
chain rooted at $A^{(1)}$ with only type $A$ players. Among all
such expansion chains, the one with the minimum number of players
(equivalently edges, since each type $A$ player occupies a single
edge) is defined below.

%Consider the expansion chain rooted at

\begin{theorem}
{\it
$EC_{min}$, the expansion chain with minimum number of edges that
supports a self-sufficient equilibrium rooted at $A^{(1)}$ is defined
by
$EC_{min} :
A^{(1)} \rightarrow
A^{(l^*+2)} \rightarrow
A^{(2l^*+3)} \rightarrow
A^{(3l^*+4)} \rightarrow
\ldots
A^{(\hC-1)}$
Every player in $EC_{min}$ has an optimal path whose length is the
maximum allowed $L^*$. The depth of chain $EC_{min}$ is $O(\hC/l^*)$.
}
\label{minEC}
\end{theorem}
For technical reasons, we don't terminate $EC_{min}$ with players
from $A^{(\hC)}$ i.e single edges of congestion 1. Such a network
can be shown to be unstable (i.e no equilibrium exists). Rather,
the optimal paths of players from $A^{(\hC-1)}$ (i.e with player
cost 4) are of length 2 with congestion 0 in $\hG$. This does not
affect our count of the total number of edges required to derive
the $PoA$ below. We need a lower bound on the number of edges to
derive an upper bound on the $PoA$, so (under)counting $EC_{min}$
only upto stage $A^{(\hC-1)}$ is acceptable for our purposes.

To prove this theorem, we need a couple of technical lemmas which
determine the minimum rate of expansion of an expansion chain. We
describe these lemmas using the preliminary setup below.
Let $\pi$ denote the set of $\hC-i+1$ players occupying a single
edge in $A^{(i)}$, for some $i \geq 1$. Let $\pi_m \in \pi$ denote
an arbitrary player with $\pi^*_m = (e_1, e_2, \ldots e_k)$ denoting
$\pi_m$'s optimal path, where $k \leq L^*$.  For the moment, assume
all edges on $\pi^*_m$ have the same congestion $c$.  We first note
that the largest stage from which type $A$ players can support
$\pi_m$ is $i + l^* +1$ since the player cost is $PC_m = 2^{\hC -i
+1}$ and we must have $k \cdot 2^{c} \geq 2^{\hC -i}$.  Using $k
\leq L^*$, we must have congestion $c \geq \hC -i - l^*$ and the
largest stage where this is possible is stage $i + l^* + 1$.
Now consider the two (partial) expansion chains
$EC_1: \pi \rightarrow A^{(i+1+l^*)} $ and
$EC_2: \pi \rightarrow A^{(i+j)} \rightarrow A^{(i+1+l^*)} $,
where $1 \leq j \leq l^* $.
We evaluate both chains at stage $i+l^*-1$.
Let $|EC_1|$ and $|EC_2|$ denote the number of edges in the respective
chains. Then we have,

\begin{lemma}
{\it $|EC_1| \leq |EC_2|$, i.e expanding directly to the $l^*+1$th succeeding
stage is cheaper than expanding via an intermediate stage.}
\label{techlemma1}
\end{lemma}
\begin{IEEEproof}
First consider $EC_1$. Since $|\pi| = \hC-i+1$ and $C^*=1$, there
are $\hC-i+1$ optimal paths at the first expansion stage of $EC_1$.
Each optimal path length is the longest allowed i.e $2^{l^*}$.
Clearly $\hC -i -l^*$ players on each edge of each such path are
enough to support the equilibrium cost of $\pi$. Thus the total
of expansion edges in $EC_1$ is $(\hC -i +1) 2^{l^*}$.

For $EC_2$, again there are $\hC-i+1$ optimal paths at the first
expansion stage.  However each edge of each optimal path now has
congestion $\hC -i -j +1$. Each optimal path must have length $l
\geq 2^{j-1}$, since $l \cdot 2^{\hC-i-j+1} \geq 2^{\hC-i}$.
Thus the total number of {\it edges} at this stage of $EC_2$ is at
least $(\hC-i+1)2^{j-1}$ while the total number of {\it players}
is at least $(\hC-i+1)(\hC -i -j +1) 2^{j-1}$.
Each of these players has its own optimal path, with each edge on a
path having congestion $\hC -i - l^*$, by definition of $EC_2$.
The cost of each optimal path must be at least $2^{\hC-i-j}$ and
so the length $l$ of each such path is at least $2^{l^*-j}$ since
$l \cdot 2^{\hC -i - l^*} \geq 2^{\hC-i-j}$.
Thus the total number of edges in this stage of
$EC_2$ is at least $(\hC-i+1)(\hC -i -j +1)2^{j-1} 2^{l^*-j}$.
Adding the edges in both stages and simplifying, we get the
overall number of edges required to support the equilibrium
of $\pi$ in $EC_2$ as
\beq
(\hC -i +1) \left [ 2^{j-1} + \frac{(\hC -i -j +1)}{2} 2^{l^*}  \right ]
\eeq
Using the fact that $\hC \geq i+j + 1$ by definition of expansion,
we can see that the number of edges in $EC_2$ is at least as much
as $|EC_1| = (\hC -i +1)2^{l^*}$.
\end{IEEEproof}

%Note that congestion is decreasing with increasing stage index $i$.

Now consider the two (partial) expansion chains $EC_3: \pi \rightarrow
A^{(i+1+l^*)} \rightarrow A^{(i+1+l^* +k)} $ and $EC_4: \pi \rightarrow
A^{(i+1+l^*-j)} \rightarrow A^{(i+1+l^* +k)} $ where $1 \leq j \leq
l^* $ and $j+k \leq l^*+1$.  (Note that the condition on $j+k$ is
because one cannot directly expand beyond $l^*+1$ stages due to the
maximum optimal path length constraint).  Then we have

\begin{lemma}
{\it $|EC_3| \leq |EC_4|$. Expanding to larger stages (i.e any stage after $i+ l^*+1$)
is cheaper via stage $i+ l^*+1$ than via any intermediate stage before it.
Equivalently (since larger stages imply expansion edges with lower congestions),
when starting from stage $i$ it is cheapest to expand via the lowest possible
congested edges which are in stage $i+ l^*+1$.
}
\label{techlemma2}
\end{lemma}
Due to space constraints, we skip the proof which counts edges
similar to the previous lemma.  The proof of  Theorem~\ref{minEC}
follows from lemmas~\ref{techlemma1}-~\ref{techlemma2}, using the
fact that starting from any stage $i$, the minimum cost expansion
arises by selecting players from stage $i+l^*+1$ to occupy expansion
edges, with all optimal path lengths being the maximum possible
$L^*$.  Due to space constraints, we omit a formal proof by induction
for showing that the number of expansion edges is minimized when
all edges on an optimal path have the same congestion.

%Intuitively, two optimal paths, one with an edge of congestion $c_k$
%and another with two edges of congestion $c_{k-1}$ have the same
%cost. However the first path must expand to $c_k$ new optimal paths
%(single edges) of congestion $c_{k-1}$. This is much larger than
%the two edges of congestion $c_{k-1}$ on the second path.

$EC_{min}$ defined in Theorem~\ref{minEC} is also the minimum sized
chain when the root players are from $B^{(1)}$ or $D^{(1)}$ although
the number of edges required in the supporting graph is slightly
different as we see later. In these cases, all stages (other than
the root) in the minimum expansion chain consist of type $A$ players
by lemma~\ref{onlyAlemma} and the proof of Theorem~\ref{minEC} is
immediately applicable in choosing the specific indices of the
expansion stages required to support the equilibrium).  As we will
show later, the $PoA$ is maximized when the chain is rooted at
$A^{(1)}$.

\begin{theorem}
{\it When $C^*=1$, the upper bound $\kappa$ on the Price of Anarchy
$PoA$ of game $\hG$ is given by the minimum of
1) $ \kappa =  O(\log L^*) $ or 2)
$\kappa \big( \log (\kappa L^*) \big) \leq \log L^* \cdot \log |E|$
\label{PoATheorem}
}
\end{theorem}
\begin{IEEEproof}
To obtain an upper bound on the $PoA$, we want to find the smallest
graph that can support an equilibrium cost of $2^{\hC}$. Since the
optimal path length $L^*$ can range from $O(1)$ to $O(|E|)$, we
evaluate smallness both in terms of path length and number of edges.

Clearly, in the case when there is no expansion in $\hG$, the Price
of Anarchy is $O(\log L^*)$, since by lemma~\ref{l^*lemma}, $\hC
\leq l_1^* + 11$ and the $PoA = \hC/C^* = O(\log L^*)$.  Consider
the case when there is expansion in the network i.e $\hC >> \log
L^*$. To bound the $PoA$, we will compute the number of edges in
the minimum sized expansion chain.  First assume there exists a
single edge of congestion $\hC$ (labeled as player set $\pi$) and
exactly one expansion chain $EC_{min}: \pi \rightarrow A^{(l^*+2)}
\rightarrow A^{(2l^*+3)} \rightarrow \ldots$ in the graph i.e the
only players in the graph are those required to be on the expansion
edges of $EC_{min}$.  Using the standard notion of depth, the node
corresponding to the player set $A^{(1+k(l^*+1))}$ on $EC_{min}$
is defined to be at depth $k$, with the root node at depth 0.  At
a given depth $k$, we define the following notations: Let $E_k$
denote the total number of expansion edges at depth $k$ (i.e the
edges on comprising the optimal paths of players at depth $k-1$),
$P_k$ denote the minimum number of players who require players from
$p_{k+1}$ on their optimal paths and $C_k$ denote the congestion
on any expansion edge.

At depth 0, we have $E_0 = 1$ (a single edge $e$ of congestion $C_0
= \hC$) Note that $P_0 = \hC -1$. Even though we have $\hC$ players,
one of these players might have its optimal path coincident with
edge $e$.  However for all $k>0$, $P_k = E_k C_k$ since all the
edges in $E_k$ are already optimal edges of players from $P_{k-1}$.
We also have $C_k = \hC - kl^* -k$ (by definition of type $A$
congestion), and finally $E_k = P_{k-1} L^*$, since every packet
in $P_{k-1}$ has its own optimal path ($C^*=1$) and every optimal
path on $EC_{min}$ is of length $L^*$.  Putting these together, we
obtain a recursive definition of $E_k = (L^*)^k P_0 \Pi_{t=1}^{k-1}
C_{t}$. We terminate our evaluation of the expansion chain when
expansion edges have a congestion of 2, i.e $\hC - kl^* -k = 2$
which implies a depth of $d = (\hC -2)/(l^*+1)$.

For technical reasons, we don't terminate the chain with players
from $A^{(\hC)}$ i.e single edges of congestion 1. Such a network
can be shown to be unstable (i,e no equilibrium exists. Rather, the
optimal paths of players from $A^{(\hC-1)}$ (i.e with player cost
4) are of length 2 with congestion 0 in $\hG$. This does not
significantly affect our count of the total number of edges required
to derive the $PoA$ below.

Thus the total number of edges in $EC_{min}$ is bounded by
\bea
|EC_{min}| &=  &1 + (\hc \- 1) \Big[ L^* +  (L^*)^2 (\hC \- l^* \-1) + \nonumber \\
& & \ \ \ldots + (L^*)^d \pi_{t=1}^{t=d} (\hC \- tl^* \-t) \Big]
\label{ecmineq}
\eea

With some algebraic manipulations, we can bound
Eq.~\ref{ecmineq} as
\beq
|EC_{min}| \geq \left( e^{-\hC} \hC^{\hC} \sqrt{\hC} (L^*)^{\hC} \right)^{\frac{1}{l^*}}
\eeq

Let $|E|$ denote the actual number of edges in graph $G$. Since
$C^*=1$, the Price of Anarchy is $\hC$. Using $\kappa$ to denote
the upper bound on the $PoA$ and simplifying, we get
\beq
\kappa( \log (\kappa L^*) -1) \leq \log L^* \cdot \log |E|
\label{mainpoa}
\eeq
Hence the $PoA$ is bounded by a polylog function of $\log |E|$ in the worst case.
\end{IEEEproof}

Can we get a larger upper bound on the $POA$ if the expansion chain
is rooted at $B^{(1)}/ D^{(1)}$ instead of $A^{(1)}$? To examine
this, let $\hC - q$ be the largest congestion in $\hG$, $q>0$. We
need $2^q$ such edges in order to satisfy the maximum player cost
of $2^{\hC}$. All these edges can be used as expansion edges for
other players.  From the analysis in Theorem~\ref{PoATheorem}, we
note that expansion between stages occurs at a factorial rate. Thus
using these $2^q$ edges as high up in the chain as possible (thereby
reducing the need for new expansion edges) will minimize the expansion
rate. The best choice for $q$ then is $l^*$. In this case, we have
a single player $\pi_m$ in equilibrium in $\hG$, occupying $L^*$
edges of congestion $\hC - l^*$. These $L^*$ edges are also the
optimal edges of $\pi_m$, i.e its equilibrium and optimal paths are
identical. Hence the first stage of expansion in this chain is for
the $L^* (\hC -l^* -1)$ players on the $L^*$ edges of $\pi_m$.  From
this point on the minimum sized chain for this graph is identical
to the minimum sized chain $EC_{min}$ defined above.  The total
number of edges in this chain can be computed in a manner similar
to above. While the number of edges is smaller than $EC_{min}$, it
can be shown that the $PoA$ is also smaller $\hC -l^*$.  Hence the
upper bound on the $PoA$ is obtained using an expansion chain rooted
at $A^{(1)}$.

So far we have assumed the optimal bottleneck congestion $C^*=1$
in our derivations.  We now show that increasing $C^*$ decreases
the $PoA$ and hence the previous derivation is the upper bound.  We
first evaluate the impact of $C^* = M > 1$ on expansion chains.
Having $C^* >1$ implies that more players can share expansion edges
and thus the rate of expansion as well as the depth of an expansion
chain (if it exists) should decrease.  We first show that expansion
chains exist even for arbitrary $C^*=M$.

\begin{lemma}
{\it Given a non-empty player set $X^{(i)} \in \{ A^{(i)}, B^{(i)},
D^{(i)} \}$, either there exists an expansion chain rooted at $X^{(i)}$
or the players of $X^{(i)}$ are on the expansion chain of other
players for all stages $i: \hC - i > 8M + l_1^* +2$. }
\label{C^*lemma}
\end{lemma}
\begin{IEEEproof}
We provide a brief outline of the proof. First consider the case
of players from $A^{(i)}$.  As before, the maximum number of optimal
edges available from within the set is $|A^{(i)}|/(\hC \- i \+ 1)$.
However each group of $M$ players could have their optimal paths
(of length one) on one such edge. Thus the number of distinct optimal
paths (edges) required is only $|A^{(i)}|/M$. If $|A^{(i)}|/M \leq
|A^{(i)}|/(\hC \- i \+ 1)$ or equivalently $\hC-i \leq M-1$, then the players
in $A^{(i)}$ are in a self-sustained equilibrium. This is not true
for the given value of $i$ in the lemma and hence there must be an
expansion chain rooted at $A^{(i)}$.
Similarly for the case of players from $B^{(i)}$, the main modification
from ~\ref{l^*lemma} is in Eq.~\ref{upperbound1} which now becomes
\beq
|B^{(i)}_0| + \sum_{j=1}^{\hc \- i \- 1} \dfrac{|B^{(i)}_j|}{2^{j-1}}
>
|B^{(i)}|/M
\label{c^*ub1}
\eeq
for making $B^{(i)}$ self-sustained since the set of $B^{(i)}$
players only need $|B^{(i)}|/M$ optimal paths.  Following the same
derivation as in lemma~\ref{l^*lemma}, Eq.~\ref{finalbound} becomes
\beq
\sum_{j=0}^{l_1^* + 2} \dfrac{|B^{(i)}_j|}{2^{j+1}}
\cdot
\big( 8M \+ j \- (\hC \- i) \big)
\geq  0
\label{C^*finalbound}
\eeq
For the given values of $i$ and $1 \leq j \leq l_1^* \+ 2$, this
is impossible and hence $B^{(i)}$ must participate in an expansion
chain. The arguments for $D^{(i)}$ are similar to lemma~\ref{l^*lemma}.
\end{IEEEproof}

Similarly Lemmas~\ref{typeB=typeA} and ~\ref{onlyAlemma} can be
suitably modified and the minimum sized chain in this case has
the same structure as defined in Theorem~\ref{minEC}. Analogous to
the $C^*=1$ case, the maximum $PoA$ occurs when $EC_{min}$ is rooted
at $A^{(1)}$. We calculate this $PoA$ with $C^*=M$, below.
\begin{theorem}
{\it When $C^*=M$, the upper bound $\kappa$ on the Price of Anarchy
$PoA$ of game $\hG$ is given by the minimum of
1) $ \kappa =  O(\frac{\log L^*}{M})$ or
2) $\kappa( \log (L^* \kappa) ) \leq \frac{l^* \log |E|}{M}$
}
\label{PoAC^*Theorem}
\end{theorem}
\begin{IEEEproof}
Suppose $\hC$ is such that there is no expansion in $G$. This implies
that $\hC \leq  8M +l_1^* +3$. The $PoA$ is $\hC/C^*$ which can be seen
to be $O(\frac{\log L^*}{M})$.
Conversely, if there is expansion we have the following:
At depth 0, $E_0=1$, $C_0 = \hC$ and $P_0 = \hC -M$ since upto $M$
players may have this edge as their optimal. As before $C_k = \hC
-kl^* -k$ and $P_k = E_k C_k$. However, now $E_k$ the number of
expansion edges at depth $k$ becomes $E_k = P_{k-1} L^* /M$ since upto
$M$ players can share the same optimal path.
Using a similar derivation as before we get,
$E_k = ( (L^*)^k/M^{k-1}) \cdot ( (\hC/M) \- 1)
\Pi_{t=1}^{k-1} C_{t}$
which after some algebraic manipulation leads to

%\bea
%|EC_{min}| &= &1 + (\frac{\hc}{M} \- 1) \Big[ L^* +   \nonumber \\
%& & \ \ \ldots + \frac{(L^*)^d}{M^{d-1}} \pi_{t=1}^{t=d} (\hC \- tl^* \-t) \Big]
%\label{ecminC^*eq}
%\eea
%After some algebraic manipulation, this yields

\beq
|EC_{min}| \geq \left( \frac{ \hC^{\hC} \sqrt{\hC} (L^*)^{\hC}}{M^{\hC} e^{\hC}}  \right)^{\frac{1}{l^*}}
\eeq
Substituting $\kappa = \hC/M$  and simplifying, we get
$l^* \log|E| \geq \hC(\log \kappa + l^* -1)$
which leads to
\beq
\kappa( \log (L^* \kappa) ) \leq \frac{l^* \log |E|}{M}
\label{C^*PoA}
\eeq
It can be seen that the $PoA$ decreases with increasing optimal
congestion $M$.
\end{IEEEproof}

%% file: sumc.pstex_t
\begin{picture}(0,0)%
\includegraphics{sumc.pstex}%
\end{picture}%
\setlength{\unitlength}{3947sp}%
\begingroup\makeatletter\ifx\SetFigFont\undefined%
\gdef\SetFigFont#1#2#3#4#5{%
  \reset@font\fontsize{#1}{#2pt}%
  \fontfamily{#3}\fontseries{#4}\fontshape{#5}%
  \selectfont}%
\fi\endgroup%
\begin{picture}(4989,6589)(806,-9842)
\put(1571,-7686){\makebox(0,0)[rb]{\smash{{\SetFigFont{12}{14.4}{\rmdefault}{\mddefault}{\updefault}$x_1$}}}}
\put(5029,-7686){\makebox(0,0)[lb]{\smash{{\SetFigFont{12}{14.4}{\rmdefault}{\mddefault}{\updefault}$y_1$}}}}
\put(2997,-7157){\makebox(0,0)[rb]{\smash{{\SetFigFont{12}{14.4}{\rmdefault}{\mddefault}{\updefault}$u$}}}}
\put(3603,-7148){\makebox(0,0)[lb]{\smash{{\SetFigFont{12}{14.4}{\rmdefault}{\mddefault}{\updefault}$v$}}}}
\put(1589,-9189){\makebox(0,0)[rb]{\smash{{\SetFigFont{12}{14.4}{\rmdefault}{\mddefault}{\updefault}$x_{k-1}$}}}}
\put(5004,-9180){\makebox(0,0)[lb]{\smash{{\SetFigFont{12}{14.4}{\rmdefault}{\mddefault}{\updefault}$y_{k-1}$}}}}
\put(1589,-8241){\makebox(0,0)[rb]{\smash{{\SetFigFont{12}{14.4}{\rmdefault}{\mddefault}{\updefault}$x_2$}}}}
\put(5021,-8232){\makebox(0,0)[lb]{\smash{{\SetFigFont{12}{14.4}{\rmdefault}{\mddefault}{\updefault}$y_2$}}}}
\put(1580,-4147){\makebox(0,0)[rb]{\smash{{\SetFigFont{12}{14.4}{\rmdefault}{\mddefault}{\updefault}$x_1$}}}}
\put(5038,-4147){\makebox(0,0)[lb]{\smash{{\SetFigFont{12}{14.4}{\rmdefault}{\mddefault}{\updefault}$y_1$}}}}
\put(3006,-3618){\makebox(0,0)[rb]{\smash{{\SetFigFont{12}{14.4}{\rmdefault}{\mddefault}{\updefault}$u$}}}}
\put(3612,-3609){\makebox(0,0)[lb]{\smash{{\SetFigFont{12}{14.4}{\rmdefault}{\mddefault}{\updefault}$v$}}}}
\put(1598,-5650){\makebox(0,0)[rb]{\smash{{\SetFigFont{12}{14.4}{\rmdefault}{\mddefault}{\updefault}$x_{k-1}$}}}}
\put(5013,-5641){\makebox(0,0)[lb]{\smash{{\SetFigFont{12}{14.4}{\rmdefault}{\mddefault}{\updefault}$y_{k-1}$}}}}
\put(1598,-4702){\makebox(0,0)[rb]{\smash{{\SetFigFont{12}{14.4}{\rmdefault}{\mddefault}{\updefault}$x_2$}}}}
\put(5030,-4693){\makebox(0,0)[lb]{\smash{{\SetFigFont{12}{14.4}{\rmdefault}{\mddefault}{\updefault}$y_2$}}}}
\put(3305,-7528){\makebox(0,0)[b]{\smash{{\SetFigFont{12}{14.4}{\rmdefault}{\mddefault}{\updefault}$p_2$}}}}
\put(3297,-3413){\makebox(0,0)[b]{\smash{{\SetFigFont{12}{14.4}{\rmdefault}{\mddefault}{\updefault}$p_1, \ldots, p_k$}}}}
\put(3305,-7016){\makebox(0,0)[b]{\smash{{\SetFigFont{12}{14.4}{\rmdefault}{\mddefault}{\updefault}$p_1$}}}}
\put(3314,-8911){\makebox(0,0)[b]{\smash{{\SetFigFont{12}{14.4}{\rmdefault}{\mddefault}{\updefault}$p_k$}}}}
\put(3305,-7989){\makebox(0,0)[b]{\smash{{\SetFigFont{12}{14.4}{\rmdefault}{\mddefault}{\updefault}$p_3$}}}}
\put(3297,-9372){\makebox(0,0)[b]{\smash{{\SetFigFont{12}{14.4}{\rmdefault}{\mddefault}{\updefault}$k-2$}}}}
\put(3314,-5820){\makebox(0,0)[b]{\smash{{\SetFigFont{12}{14.4}{\rmdefault}{\mddefault}{\updefault}$k-2$}}}}
\put(3305,-6230){\makebox(0,0)[b]{\smash{{\SetFigFont{12}{14.4}{\rmdefault}{\mddefault}{\updefault}Nash Equilibrium with social cost $k$}}}}
\put(3297,-9782){\makebox(0,0)[b]{\smash{{\SetFigFont{12}{14.4}{\rmdefault}{\mddefault}{\updefault}Routing with optimal social cost 1}}}}
\end{picture}%